%% file: main.tex
\documentclass{article}
\pdfoutput=1
     \usepackage[nonatbib,preprint]{neurips_2020}

\usepackage[utf8]{inputenc} % allow utf-8 input
\usepackage[T1]{fontenc}    % use 8-bit T1 fonts
\usepackage{hyperref}       % hyperlinks
\usepackage{url}            % simple URL typesetting
\usepackage{booktabs}       % professional-quality tables
\usepackage{amsfonts}       % blackboard math symbols
\usepackage{nicefrac}       % compact symbols for 1/2, etc.
\usepackage{microtype}      % microtypography
\input{preamble}

\title{Interpretable Super-Resolution \\ via a Learned Time-Series Representation}

\author{%
  Randall~Balestriero \\
  ECE Department\\
  Rice University, USA
   \And
   Hervé~Glotin \\
   Univ. Toulon, Aix Marseille Univ.\\ CNRS, LIS, France
   \And
  Richard~G.~Baraniuk \\
  ECE Department\\
  Rice University, USA
}
\usepackage{enumitem}
\begin{document}

\maketitle

\begin{abstract}
    \input{abstract}
\end{abstract}

\input{introduction}

\input{background}
\input{invariance}

\input{conclusions}

\section*{Broader Impact}

In this paper we propose a novel learnable super-resolution time-frequency resolution. The algorithm adapts to the data and task at hand and is thus not specialized to any type of settings. We validated the method on various public dataset open for academic purposes. We also highlighted the key properties of the transform that are general. Finally the interpretability of the method should make this paper of great interested for anyone aiming to have a end to end pipeline on time-series.

\section*{Acknowledgments}
This work was partly granted by 
NSF grants CCF-1911094, IIS-1838177, and IIS-1730574; 
ONR grants N00014-18-12571 and N00014-17-1-2551;
AFOSR grant FA9550-18-1-0478; 
DARPA grant G001534-7500; and a 
Vannevar Bush Faculty Fellowship, ONR grant N00014-18-1-2047 for Randall Balestriero and Richard Baraniuk, and partly granted by
ANR-18-CE40-0014 SMILES, ANR ADSIL Artificial Intelligence Chair, IUF, and MI CNRS MASTODONS (\url{http://sabiod.org}) for Herve Glotin.

\small

\bibliography{BIBLIO}
\bibliographystyle{plain}

\clearpage
\appendix

\input{proof}

\end{document}

%% file: preamble.tex
\usepackage{latexsym,amsmath,amssymb,amsthm,eucal,bbm,color}
\usepackage{url}
\usepackage{comment}
\usepackage{float}
\usepackage{tcolorbox}
\usepackage{booktabs}
\usepackage{blindtext}
\usepackage{hyperref}
\hypersetup{
    colorlinks,
    linkcolor={red!50!black},
    citecolor={blue!50!black},
    urlcolor={blue!80!black}
}
\usepackage{multirow}
\usepackage{booktabs}
\usepackage{array}
\usepackage{soul}
 
\usepackage[normalem]{ulem}
\usepackage{stackengine}
\usepackage{parskip}
\usepackage{bm}
\usepackage{balance}
\usepackage{comment}
\usepackage{soul}
\usepackage{listings}
\usepackage[version=4]{mhchem}
\usepackage{empheq}
\usepackage{mdframed}
\usepackage{nomencl,etoolbox,ragged2e,siunitx}
\usepackage{tikz}
\usetikzlibrary{shapes,arrows}
\usetikzlibrary{er,positioning}

\usetikzlibrary{fit,positioning}
\tikzstyle{block} = [rectangle, draw, fill=blue!20, 
    text width=12.8em, text centered, rounded corners, minimum height=4em]
\tikzstyle{line} = [draw, -latex']
\tikzstyle{cloud} = [draw, ellipse,fill=red!20, node distance=3cm,
    minimum height=2em]
\usepackage{mdframed}

\newtheorem{thm}{Theorem}

\newtheorem{prop}{Proposition}

\newtheorem{lemma}{Lemma}
\newtheorem{defn}{Definition}

\newmdtheoremenv{mhyp}{Hypothesis} 
\newmdtheoremenv{mthm}{Theorem}
\newmdtheoremenv{mtheorem}{Theorem}
\newmdtheoremenv{mprop}{Proposition}
\newmdtheoremenv{mcor}{Corollary}
\newmdtheoremenv{mlemma}{Lemma}
\newmdtheoremenv{mdefn}{Definition}
\newmdtheoremenv{mmydef}{Definition}
\newmdtheoremenv{mconj}{Conjecture}
\newmdtheoremenv{mex}{Example}
\newmdtheoremenv{mexercise}{Exercise}
\usepackage{wrapfig}
\usepackage{caption}

\usepackage{tikz}
\usetikzlibrary{calc}
\usepackage{pgfplots}
\usepackage{mwe}
\usepackage[]{algorithm2e}
\usepackage{cancel}
\DeclareMathAlphabet\mathbfcal{OMS}{cmsy}{b}{n}

\def \C{\text{C}}

\def \F{\mathcal{F}}
\def \K{\text{K}}

\allowdisplaybreaks

\def \bx{\boldsymbol{x}}

\newcommand{\WV}[0] {{\text{WV}}}
\newcommand{\WT}[0] {{\text{WT}}}
\newcommand{\ST}[0] {{\text{ST}}}
\newcommand{\SP}[0] {{\text{SP}}}

\newcommand{\STFT}[0] {{\text{STFT}}}

\newcommand{\qq}{\vspace*{-2mm}}

\usepackage{xcolor}

\usepackage{arydshln}
\usepackage{tikz}
\usetikzlibrary{shapes,arrows}
\usetikzlibrary{er,positioning}

\usetikzlibrary{fit,positioning}
\tikzstyle{block} = [rectangle, draw, fill=blue!20, 
    text width=12.8em, text centered, rounded corners, minimum height=4em]
\tikzstyle{line} = [draw, -latex']
\tikzstyle{cloud} = [draw, ellipse,fill=red!20, node distance=3cm,
    minimum height=2em]

%% file: abstract.tex
We develop an interpretable and learnable Wigner-Ville distribution that produces a super-resolved quadratic signal representation for time-series analysis.
Our approach has two main hallmarks. 
First, it interpolates between known time-frequency representations (TFRs) in that it can reach super-resolution with increased time and frequency resolution beyond what the Heisenberg uncertainty principle prescribes and thus beyond commonly employed TFRs, 
Second, it is interpretable thanks to an explicit low-dimensional and physical parameterization of the Wigner-Ville distribution.
We demonstrate that our approach is able to learn highly adapted TFRs and is ready and able to tackle various large-scale classification tasks, where we reach state-of-the-art performance compared to baseline and learned TFRs.

%% file: introduction.tex
\section{Introduction}

% Time-series appear in a tremendous variety of fields such as health \cite{addison2000evaluating,miwakeichi2004decomposing}, bioacoustics \cite{neti2000audio,stowell2015detection} speech \cite{dupont2000audio,oord2016wavenet}, geophysics \cite{jupp1975stable,menke2018geophysical}, economics \cite{kendall1953analysis,granger2014forecasting}, astrophysics \cite{scargle1979studies,scargle1982studies}; and in a variety of tasks ranging from compression \cite{srinivasan1998high}, data exploration \cite{cohen1995time} to classification \cite{joly2016lifeclef} or anomaly detection \cite{seydoux2016detecting}.

With the recent deep learning advances \cite{lecun2015deep,goodfellow2016deep} there has been an exponential growth in the use of Deep Networks (DNs) on various time-series. However, the vast majority of DNs do not directly observe the time-series data but instead a handcrafted, a priori designed representation. Indeed, the vast majority of state-of-the-art methods combine DNs with some variant of a {\em Time-Frequency Representation} TFR \cite{lattner2019learning,purwins2019deep,liu2019bottom}. A TFR is an {\em image} representation of a time-serie obtained by convolving the latter with a filter-bank, such as wavelets or localized complex sinusoids (e.g., Gabor transform).
Different filter-banks lead to different TFR families. 

The TFR-DN combination is powerful due to three major reasons:  (i) the TFR contracts small transformations of the time-serie internal events such as translation, time and/or frequency warping \cite{bruna2013invariant} leading to more stable learning and faster DN convergence; (ii) the image representation allows to treat TFRs as a computer vision task where current DNs excels; (iii) given a coherent choice of TFR, the representation of the features of interests, such as phonemes for speech \cite{waibel1989phoneme}, form very distinctive shapes in the TFR image, with dimensionality much smaller than the event's time-serie representation. In fact, a single coefficient of the TFR can encode information of possibly thousands of contiguous bins in the time-serie representation \cite{logan2000mel}.
Those TFR benefits allow for great performance gains in a vast majority of task and dataset. Nevertheless, the choice of TFR has the potential to dim, or amplify, the above benefits further pushing the performance gains.

Choosing the ``best'' TFR is a long lasting research problem in signal processing \cite{coifman1992entropy,jones1994simple,donoho1994minimum}. 
While TFR selection and adaptation was originally driven by signal reconstruction and compression, the recent developments of large supervised time-series datasets have led to novel learnable solutions that roughly fall into four camps.
First, methods relying on the Wavelet Transform (WT) \cite{meyer1992wavelets}. A WT is TFR with a constant-Q filter-bank based on dilations of a mother wavelet. In \cite{pmlr-v80-balestriero18a} the learnability of the mother wavelet is introduced by means of cubic spline parameterization of the mother wavelet and learning the shape by learning the values of the spline and its derivative at the knots position. 
Second, methods relying on band-pass filters without center-frequency to bandwidth (Q) constraint such as the Short-Time Fourier Transform (STFT) \cite{allen1977short}. The first method \cite{NIPS2018_7711} performs this by independently learning the center frequencies and bandwidths of a collection of Morlet wavelets (learning of those coefficients intependently breaks the constant-Q property). Another method \cite{ravanelli2018interpretable} relies on learning the start and cutoff frequency of a bandpass sinc filter apodized with an hamming window. Those methods are thus similarly learning the location of the bandpass but use a different apodization window of a complex sine (Gaussian or hamming).
Third, \cite{zeghidour2018learning} proposes to learn Mel filters that are applied onto a spectrogram (modulus of STFT). Those filters linearly combine adjacent filters in the frequency axis which can be interpreted as learning a linear frequency subsampling of the spectrogram; learning the apodization window used to produce the spectrogram has also been developped in \cite{jaillet2007time,pei2012stft}.
Finally, there are also methods relying on unconstrained DN layers applied on the time-series but which are pre-trained such that the induced representation (layer output) resembles an a priori determined target TFR. This has been done for chirplet transforms\cite{baraniuk1996wigner} in \cite{glotin2017fast} and for Mel-Spectrograms in \cite{ccakir2018end}.

All the above methods for choosing the ``best TFR'' for a DN application suffer from at least one of the three following limitations: (i) the inability to interpolate between different TFR families due to family specific parameterization of the learnable filter-banks; (ii) the inability to maintain interpretability of the filter-bank/TFR after learning; (iii) the inability to reach super-resolution in time and frequency to allow more precise representation.
There is thus a need to {\em provide a universal learnable formulation able to interpolate between and within TFRs while preserving interpretability of the learned representation and with the ability to reach super-resolution}. We propose such a representation in this paper by a specific parametrization of the Wigner-Ville Distribution (WVD) \cite{wigner1932,flandrin1998time}.
We validate our method on multiple large scale datasets of speech, bird and marine bioacoustic and general sound event detection, and demonstrate that the proposed representation outperforms current learnable TFR techniques as well as fixed baseline TFRs when combined with various DNs. 

Our contributions are:

{\bf [C1]}
We develop a Wigner-Ville Distribution based TFR with explicit interpretable parameterization that can interpolate between any known TFR and reach super-resolution (Sec.~\ref{sec:transform}) and derive various properties such as invariance and covariance of the representation and its stability to parameter perturbations (Sec.~\ref{sec:invariance}).

{\bf [C2]}
We provide an efficient implementation allowing to compute the proposed representation solely by means of Short-Time Fourier Transforms (Sec.~\ref{sec:STFT}). This allows GPU friendly computation and applicability of the method to large scale time-series dataset. We study the method complexity and provide a detailed pseudo-code and Python implementation (Sec.~\ref{sec:complexity}).
The Python code is available at \url{TBD.git}.

{\bf [C3]}
We validate our model and demonstrate how the proposed method outperforms other learnable TFRs as well as fixed expert based transforms
on various dataset and across multiple DN architectures. We interpret the learned representations hinting at the key features of the signals needed to solve the task at hand (Sec.~\ref{sec:experiments}).

% we now review the TFRs and in particular the WVD onto which our representation is based.

%% file: background.tex
\section{Background on Time-Frequency Representations}
\label{sec:back}

%We recall some basic definitions and properties of TFRs.

{\bf Fourier and Spectrogram.}~
Motivated by the understanding of physical phenomena, mathematical analysis tools have been created, notably the Fourier transform \cite{bracewell1986fourier}. This representation corresponds to a signal expansion into the orthogonal family of complex exponentials as
$
    \F_{x}(\omega)=\int_{-\infty}^{\infty}x(t)e^{-i \omega t}dt,
$
which provides a powerful representation for stationary signals.
In many situations, it does not seem reasonable to assume the entire series to be stationary.
For non stationary signal analysis, where the observed signal carries different information throughout its duration, co-existence of the time and frequency variables in the representation is needed. 
One solution is offered by the Short Time Fourier Transform (STFT) \cite{allen1977short} defined as follows
\begin{align}
        \STFT_{x, w}(t,f) = \int_{-\infty}^{\infty} w(t-\tau)x(\tau)e^{-i f \tau}d\tau,
\end{align}
with $w$ an apodization window which vanishes when moving away from $0$. This representation thus only assumes stationarity within the effective support of $w$ which we denote by $\sigma_{\rm T}$. The squared modulus of the STFT is called the spectrogram as 
$
	\SP_{x}(t,\omega)=| ST_{x}(t,\omega) |^2.
$
The simplicity and efficiency of its implementation makes the spectrogram one of the most widely used TFR for non-stationary signals. Nonetheless, the spectrogram has a fundamental antagonism between its temporal and frequency resolution which depends on the apodization window spread $\sigma_{\rm T}$.
In a spectrogram, large $\sigma_t$ allows high frequency resolution and poor time resolution and conversely for small $\sigma_t$. Most applications employ a Gabor (or truncated/approximated) apodization window, in which case $w(u)=\frac{1}{\sqrt{2 \pi}\sigma_t}e^{-t^2/(2\sigma_t^2)}$. Such spectrograms are denoted as Gabor transforms \cite{gabor1946theory}. We will denote such a Gaussian window by $g_{\sigma_t}$.

{\bf Wavelet Transform.}
~
A wavelet filter-bank $B$ is obtained by dilating a mother filter $\psi_0$ with various {\em scales} $s$ which correspond to frequencies $f$ via $s=2^{S(1-\frac{f}{\pi})}$ with $S$ the largest scale to be analyzed. Application of those dilated filters onto a signal leads to the wavelet transform WT \cite{mallat2008wavelet}
\begin{align*}
        \WT_{x}(t,s)&=(x \star\psi_{s})(t), \text{ where } \psi_{s}(t) =\frac{1}{\sqrt{s}} \psi_0\left(\frac{t}{s}\right),
\end{align*}
with $s > 0$. The relative position of the mother wavelet center frequency is not relevant as the scales can be adapted as desired, let consider here that $\phi_0$ is placed at the highest frequency to be analyzed.
% , B = \left \{ \psi_{\lambda}, \lambda  \in \{ 2^{\frac{i}{Q}} ,i=0,\dots,JQ-1\}\right \}, \;\;\;\text{ where } \psi_{\lambda}(t) = \psi_0\left(\frac{t}{\lambda}\right)
As opposed to the spectrogram, the resolution of the WT varies with frequencies as the filters have a constant bandwidth to center frequency ratio.  In a WT, high frequency atoms $\psi_{\lambda}$ with $\lambda$ close to $1$  are localized in time offering a good time resolution but low frequency resolution. Conversely, for low frequency atoms with $\lambda \gg 1$, the time resolution is poor but the frequency resolution is high. As natural signals tend to be of small time duration when they are at high frequencies and of longer duration at low frequencies \cite{daubechies1990wavelet}, the scalogram is one of the most adapted TFR for natural biological signals \cite{meyer1992wavelets}.

{\bf Wigner-Ville Transform and Cohen Class.}~
The Wigner-Ville (WV) transform (or quasi probability or distribution) \cite{wigner1932} 
describes the state of a quantum particle by a quasi-probability distribution in the phase space formulation of quantum mechanics instead of being described by a vector on the Hilbert space \cite{moyal1949quantum}. It was then leveraged outside of quantum physics as a quadratic time frequency representation for signal processing \cite{ville1948theorie}. The transform combines complex sine filters as the Fourier transform and auto-correlations of the signal as follows
\begin{align}
	\WV_{x}(t, \omega) = \int_{-\infty}^{\infty} x\left(t-\frac{\tau}{2}\right)
	x^*\left(t+\frac{\tau}{2}\right)e^{-i \omega \tau}d\tau.
	\label{eq:WV}
\end{align}

Due to its form, computing the WV is demanding as it provides a representation that is highly dimensional in time and in frequency. Nevertheless, and as opposed to the spectrogram or scalogram, the WV has perfect time and frequency resolution. 
%This comes at the cost of introducing interference in the representation between different events of the signal as a byproduct of the uncertainty principle . 
To focus on specific physical properties and obtain different representations, it is possible to convolve the WV with a $2$-dimensional low pass filter $\Pi$ as :
\begin{align}
    \C_{x}(t, f; \Pi) = (\WV_{x} \star \Pi)(t,f),
    \label{eq:cohen}
\end{align}
where $\star$ represents the $2$-dimensional convolution. This filtering lessen the time and frequency resolution while removing the present interferences due to the auto-correlation of the different events present in the signal. The collection of representations  $\{C_{x}(.,.;\Pi), \Pi \in \mathbb{L}^2(\mathbb{R}^2), \Pi \geq 0 \}$ defines the {\em Cohen class} \cite{cohen1989time}. For example, any spectrogram lives in this class, as well as many other known TFRs. 

This paper extends this representation to allow learning and computational efficiency while providing interpretability of the produced representation.

%% file: invariance.tex
\section{Learnable Wigner-Ville Based Signal Representation}

In this section we develop our proposed signal representation which builds upon the Wigner-Ville Distribution. We first define our representation and study some key properties to finally propose a physics based parameterization that will allow for interpretability and robust learning.

\subsection{Interpretable, Universal and Stable Signal Representation}
\label{sec:transform}

We now define our transformation, coined the {\em $K$-transform}, which corresponds to applying a kernel $\Phi[t,f] \in \mathbb{L}^2(\mathbb{R} \times \mathbb{R}^+)$ onto the Wigner-Ville transform $\WV_{x} \in \mathbb{L}^2(\mathbb{R} \times \mathbb{R}^+)$ of the signal $x \in L^2(\mathbb{R})$ (recall (\ref{eq:WV})), for each time $t \in \mathbb{R}$ and frequency $f\in \mathbb{R}$ as follows.
This transform corresponds to replacing the time-frequency-invariant convolution of Cohen's class (recall (\ref{eq:cohen})) with a more general time-frequency-varying convolution.

\begin{defn}
The $K$-transform of a signal $x$ with kernel $\Phi$ is defined as 
\begin{align}
\K_{x,\Phi}(t,f) = \left \langle \WV_{x} , \Phi[t, f]\right \rangle.\label{eq:transform}
\end{align}
\end{defn}

The $K$-transform is thus defined by taking an inner product of the Wigner-Ville representation with a $2$-dimensional filter given by $\Phi[t, f]$ for each time $t$ and frequency $f$. It is real valued as the Wigner-Ville representation is real and so is the kernel $\Phi$.
In particular, we are interested at an interpretable kernel $\Phi$ and thus propose to parameterize it as a $2$-dimensional Gaussian function given by
\begin{gather}
\Phi[t,f](\tau, \omega)=\mathcal{N}\left([\tau, \omega]^T; \mu(t,f), \Sigma(t,f)\right),
\end{gather}
where $\mathcal{N}$ the Gaussian multivariate function \cite{degroot2012probability}; we also explicit the mean and covariance parameters given by 
\begin{gather}
\mu(t,f)=\left(\mu_{\rm t}(t, f), \mu_{\rm f}(t, f)\right)^T,\Sigma(t,f)=\begin{pmatrix}
    \sigma_{\rm t}(t, f)^2 & \rho(t,f)\\
    \rho(t,f) & \sigma_{\rm f}(t, f)^2
\end{pmatrix}.\label{eq:mu_cov}
\end{gather}

Such parametrization of the kernel employed onto the WV arises naturally in various contexts as it corresponds (through some specific mean and covariance configurations) to the analytical kernel associated to known TFRs such as the chirplet transform, the Wavelet with Morlet wavelet or the Gabor transform 
\cite{nuttall1988wigner,flandrin1990affine,jeong1990variable,baraniuk1996wigner,talakoub2010approximating}. Second, the explicit parameterization of the mean and covariance matrix (\ref{eq:mu_cov}) allows to interpret those coefficients as to what are the physical properties of the kernel.

{\bf Interpretability.}~
A great advantage of the Gaussian formulation and explicit parameterization from (\ref{eq:mu_cov}) resides in its simplicity to be interpreted and analyzed. In fact, we directly obtain that the vector $\mu(t,f)$ encodes the center frequency and time position information. The covariance matrix encodes the frequency and time bandwidths of the filter as well as the covariance coefficient $\rho$ which encodes the chirpness (how instantaneous frequency of the filter changes with time). Due to the low dimensionality of the parameter vector
\begin{align}
 \theta(t,f)=(\mu_{\rm t}, \; \mu_{\rm f}, \; \sigma_{\rm t}, \; \sigma_{\rm f}, \; \rho)(t,f),\label{eq:theta}
\end{align} 
comparing learned representations can also be done in a straightforward manner by comparing $\theta(t,f)$ and the other parameter realization $\theta'(t,f)$. In fact, we now demonstrate how the representation is contractive w.r.t. those parameters, that is, how close parameters imply close K-transforms. Let denote directly the representation obtained with the Gaussian kernel $\K_{x,\theta}$.

\begin{lemma}
\label{cor:lipschitz_theta}
Given two parameter vectors $\theta(t,f)$ and $\theta'(t,f)$, the distance between the two induced representations is upper bounded by the distance between their parameters as
\begin{align*}
    \| \K_{x,\theta} - \K_{x,\theta'} \|_{L^2(\mathbb{R}^2)}^2 \leq \kappa \|x \|^2_{L^2(\mathbb{R})}  \left(\int_{t,f}\|\theta(t,f)-\theta'(t,f) \|_{2}^2\right)^{-\frac{1}{2}}.
\end{align*}
with $\kappa$ the Lipschitz constant of a standard 2D Gaussian ($\approx 0.2422$) (proof in App.~\ref{proof:lipschitz_theta}).
\end{lemma}
\qq

In particular, for unit norm signal $x$ the representation will be contractive w.r.t. the $\theta$ parameters.
As a result, given a fixed signal $x$, close parameters $\theta$ and $\theta'$ imply close representations $\K_{x,\theta}$ and $\K_{x,\theta'}$ and thus it is sufficient to compare the vectors of parameters between learned representation or between learned and known representation to assert on what information is encoder and how as studied in Table~\ref{tab:summary_parameters}.

{\bf Universality of the Representation and Continuous Interpolation.}~
We now provide the following result that demonstrates how the formulation from (\ref{eq:transform}) is sufficient to span any common TFR. In particular, due to the Gaussian parametrization, not all TFR are reachable by varying $\theta$, we thus state the following result for our special case but prove the more general case of arbitrary kernel $\Phi$ in the proof.

\begin{figure}
    \centering
    \begin{minipage}{0.03\linewidth}
    \rotatebox{90}{frequency ($f$)}
    \end{minipage}
    \begin{minipage}{0.34\linewidth}
    \centering
    \includegraphics[width=1\linewidth]{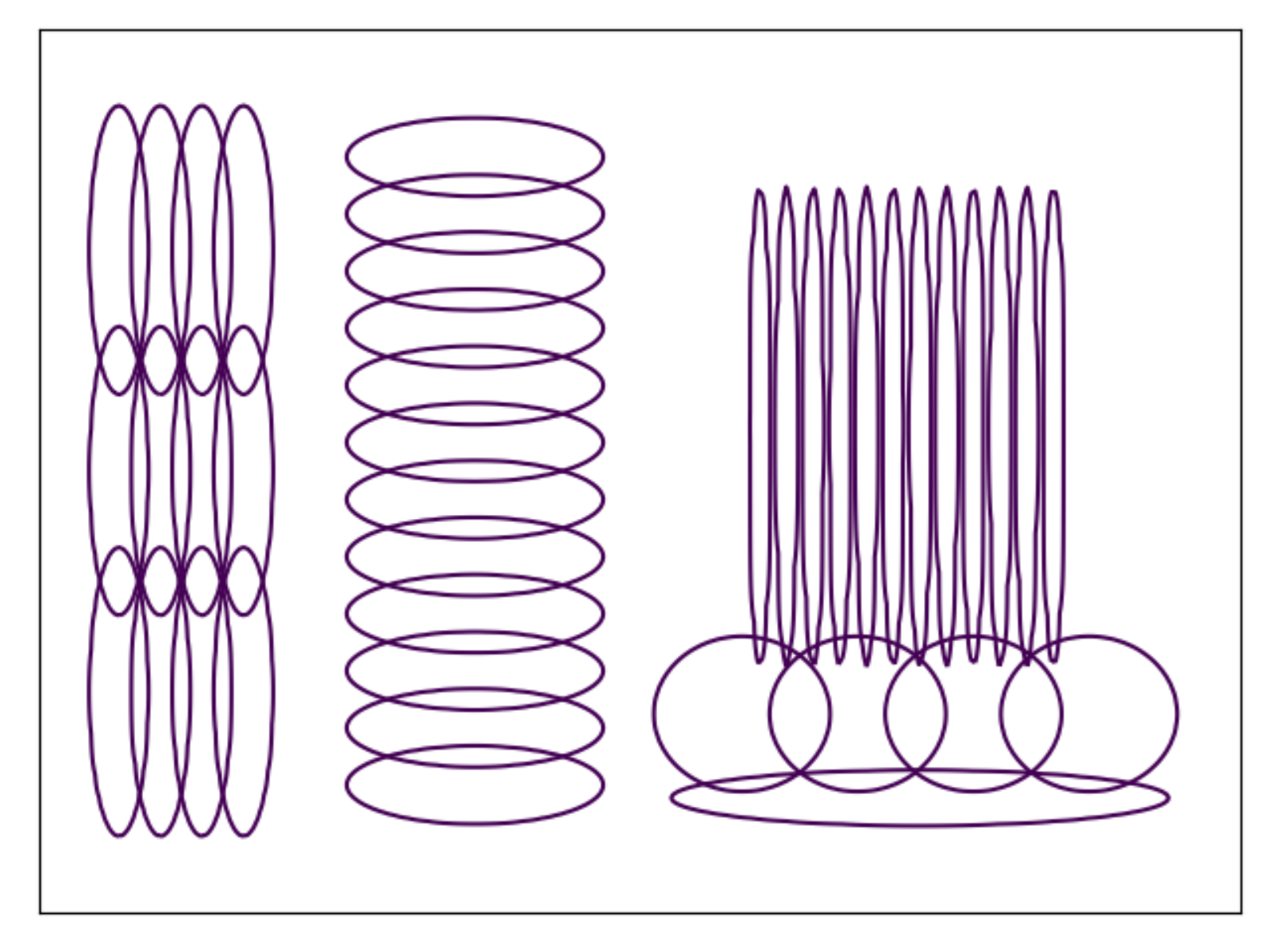}\\
    time ($t$)
    \end{minipage}
    \begin{minipage}{0.6\linewidth}
    \caption{{\small Example of Gaussian parameterization leading to a STFT with high time resolution and poor frequency resolution (left) or STFT with poor time resolution and high frequency resolution (middle) and finally the wavelet (or constant Q-transform) case with adaptive time frequency resolution, for a chirplet transform consider this last case with a non diagonal covariance matrix, with chirpness $\rho$, the support of each effective Gaussian is depicted. The $\K$-transform uses a 2D-Gaussian  for computing each $\K(t,f)$ coefficients, with learnable mean and covariance, thus allowing to continuously interpolate between known TFR or learn a new one based on the data and task at hand.
    }}
    \label{fig:example}
    \end{minipage}
    \vspace*{-4mm}
\end{figure}

\begin{prop}
\label{prop:universal_gabor}
Any Gabor transform, Gabor wavelet transform, Gabor chirplet transform (with arbitrary chirpness, from $\rho(t,f)$) and their signal adapted variants can be reached by the $\K$-transform.
(proof in App.~\ref{proof:universal_gabor}).
\end{prop}
\qq

The above result is crucial to understand that it is not only possible to reach various TFRs such as constant-Q transforms or spectrograms with Gaussian apodization window but also any representation that has been adapted to the signal \cite{abramovich1998wavelet}, that is, with a time-dependent set of filter-bank. We depict in Fig.~\ref{fig:example} some examples of Gaussian parameterization. To explicit the above let consider one case of signal based representation given by wavelet packet trees \cite{ramchandran1993best}. In this formulation, the representation we denote as $P_{x}$ will adapt for each time and frequency which wavelet family with the given translation and dilation is used to obtain $P_{x}(t,f)$ such that the final representation $P_{x}$ is optimal in some sense, for example with minimal Entropy. In this case, the final representation is adapted to the signal with filters varying with time and not just with frequencies. Such representation do not belong in the Cohen class but are reachable with the proposed formulation.

Another crucial property that is beneficial to learning and stability of the representation is to ensure that when moving the parameter $\theta$ from one TFR to another, the representation moves continuously and regularly which is formalized in the next result following directly from Cor.~\ref{cor:lipschitz_theta}.

\begin{thm}
\label{thm:continuity}
The representation $\K_{x ,\Phi}$ moves continuously with $\Phi$ and allow to continuously interpolate between any desired TFR (proof in App. \ref{proof:continuity}).
\end{thm}
\qq

As a result, the proposed representation will be stable during learning where the kernel $\Phi$ (or $\theta$) will be incrementally updated based on some rule. Continuity of the representation combine with contractivity from Cor.~\ref{cor:lipschitz_theta} ensure that small updates of the parameters will also lead to small updates in the representation and thus not lead to unstable learning/gradient based updates.

\subsection{Equivariance and Invariance Properties}
\label{sec:invariance}

Recall that the K-transform does not impose any dependency between different time nd frequency kernels $\Phi[t,f]$ as opposed to most WV based transforms such as the affine Cohen class \cite{flandrin1990affine,flandrin1993temps,daubechies2002adaptive}, the Generalized Cohen Class \cite{janssen1982locus}, Pseudo Wigner Ville \cite{flandrin1984interpretation}, Smoothed Pseudo Wigner-Ville \cite{hlawatsch1995smoothed}. Due to this freedom, the K-transform is not inherently equivariant to signal translation and/or frequency shift. Nevertheless, we now demonstrate that it is straightforward to impose some standard properties onto $\K_{x,\Phi}$ as follows where we denote the time convolution as $\star_{\rm t}$ and the frequency convolution as $\star_{\rm f}$.

\begin{prop}
\label{prop:equivariant}
The K-transform can be equipped with the following properties: (i) translation equivariance  $\iff \Phi[t-\tau,f]=\Phi[t,f](.-\tau,.) \iff \K_{x,\Phi}(.,f) = \WV_{x} \star_{\rm t} \Phi[0,f]$, (ii) frequency shift equivariance $\iff \Phi[t,f-\omega]=\Phi[t,f](.,.-\omega) \iff \K_{x,\Phi}(t,.) = WV_{x} \star_{\rm f} \Phi[t,.]$ (proof in \ref{proof:invariant}).
\end{prop}
\qq

As a result, for the K-transform to be equivariant to translation and frequency shifts the kernel $\Phi$ must be such that $\Phi[t,f](\tau, \omega)=\Phi[t',f'](\tau-(t'-t),\omega-(f'-f))$ or equivalently, the representation can be written as the $2D$ convolution of the Wigner-Ville transform with the $2D$ kernel $\Phi[0,0]$ as $
\K_{x,\Phi} = \WV_{x} \star \Phi[0,0]$
which falls back to the Cohen class (recall (\ref{eq:cohen})) and in particular the Smoothed Pseudo Wigner Ville \cite{andria1996interpolated}. Notice that if the $\K$-transform is constrained to only be translation equivariant, and that the covariance matrices $\Sigma(t,f)$ are diagonal with elements (constant in time) varying with frequency based on the center frequencies $\mu_f(t,f)$ then the $\K$-transform falls back to the Affine Smoothed Pseudo Wigner Ville \cite{flandrin1990affine}.

\iffalse

On the other hand, invariance is a strong property that would impose $\Phi$ to be a constant across time and/or frequency. To provide further insights, we consider local invariance, that is, the representation is unchanged only for small translations and/or frequency shifts. We thus have the following result.

In fact, depending in the context and task, having invariant representations can be seen as preferred if followed by other algorithms such as Deep Neural Networks which can then learn based on the data the optimal invariants if not known a priori \cite{mallat2016understanding}. In fact, equivariance allows to preserve the information present in the signal w.r.t. each transformation operator such as translation. On the other hand, invariance can also be imposed a priori if one has knowledge of the information that should be seen as uninformative for the task at hand.

\fi

A key property of most TFR reside in their time-frequency resolution, that is, how precise will be the representation into reflecting the frequency content at each tiem step present in the studied signal. Standard TFR have limited resolution while the WV has perfect time and frequency resolution.
The $\K$-transform can easily reach super resolution by learning small covariance matrices $\Sigma(t,f)$, however, we now demonstrate that this super-resolution will make the representation more sensitive to input perturbations, and vice-versa. Denote a transformation $D$ applied on the signal $x$ and formally characterize the induced perturbation amount in the $K$-trasnform by $\| \K_{x, \Phi}-\K_{D(x),\Phi} \|_{L^2(\mathbb{R}^2)}$ \cite{mallat2012group}, let also denote $\det(\Sigma(t,f))=\sigma_{\rm t}\sigma_{\rm f}-\rho^2$.

\begin{prop}
\label{prop:invariance}
A transformation of the signal $D(x)$ implies a change in the representation proportional to the inverse of $\sigma_{\rm t}\sigma_{\rm f}-\rho^2$ as
$
\|\K_{x,\Phi}-\K_{D(x),\Phi} \|_{L^2(\mathbb{R}^2)} \leq \max_{t,f}\frac{\kappa}{\det(\Sigma(t,f))} \times \|x-D(x) \|_{L^2(\mathbb{R})},
$
with $ \kappa$ the Lipschitz constant of the WV which exists and is finite for bounded domain
(proof in App.~\ref{proof:invariance})
\end{prop}
\qq

Based on the data and task at hand, the $\K$-transform can adapt its parameters to obtain an invariant representation a la scattering network \cite{mallat2012group} thus providing a robust time-serie representation, while having poor resolution, or, conversely, it can reach a super-resolution representation.

We now demonstrate how the $K$-transform can be computed efficiently solely from Fast Fourier Transforms, for other properties of the transform such as characterization of the interference based on the covariance please see Appendix~\ref{sec:interference}.

\section{Fast Fourier Transform Computation}

The WV, the seed of any possible K-transform, for a discrete signal of length $N$ is a $N\times N$ matrix obtained by doing $N$ Fourier transforms of length $N$. Its computational complexity is thus quadratic ($\mathcal{O}(N^2\log(N))$) making it unsuited for large scale tasks. We now demonstrate how to leverage the Gaussian parameterization to greatly speed-up the K-transform computation.

\subsection{The Short-Time-Fourier-Transform Trick}
\label{sec:STFT}

In order to provide a fast implementation, we will first draw the link between a spectral correlation of the STFT and the WV from which we extend our method.
Denote the STFT of a signal with a Gabor apodization window of time spread $\sigma$ by $\STFT_{x,\sigma}$. We now demonstrate the equivalence between (i) doing a spectral autocorrelation of $\STFT_{x,\sigma}$ and (ii) doing a 2D Gaussian with diagonal covariance time-frequency-invariant convolution of the WV.

\begin{lemma}
\label{thm:fast}
The spectral correlated STFT corresponds to a smoothed Wigner-Ville Distribution  as
\begin{multline}
    \left(\WV_{x} \star \mathcal{N}\left(.;(0,0)^T, diag(
\sigma_{\rm t}, \sigma_{\rm f})\right)\right)(t,f)=\int_{-\infty}^{\infty} 
g_{\sigma_{\rm f}^{-1}-\sigma_{\rm t}}(\omega)e^{j 2 \pi \omega t}
 \STFT _{x,\sigma_{\rm f}^{-1}}(t,f+\frac{\omega}{2})\\\times \STFT_{x,\sigma_{\rm f}^{-1}}^*(t,f-\frac{\omega}{2})d\omega,\label{eq:fast}
\end{multline}
where $g_{\sigma}$ is a $1$-dimensional Gaussian function with spread $\sigma$ and $\sigma_{\rm t}\leq \sigma_{\rm f}^{-1}$. (Proof in App. \ref{proof:fast}).
\end{lemma}
\qq

The above is very fast to be computed in practice for finite sample signals as long as $\sigma_{\rm f}^{-1}$ is small as this defines the largeness of the window of the STFT. Similarly, whenever $\sigma_{\rm f}^{-1}-\sigma_{\rm t}$ is small, computing (\ref{eq:fast}) is fast as it reduces the amount of frequencies to go through during the frequency integral computation. It is now possible to leverage this already convolved (and fast) transform to obtain the K-transform. Notice that any $\K$-transform can be reached that way, as long as $2D$ Gaussian used in (\ref{eq:fast}) has smaller effective support than the 2D Gaussians of the desired $\K$-transform.

\begin{thm}
\label{thm:gaussianproduct}
Any K-transform (recall (\ref{eq:transform}) can be obtained from convolving the WV as follows
\begin{align*}
    K_{x,\theta}(t,f)=\left(\left(\WV_{x} \star \mathcal{N}\left(.;(0,0)^T, diag(
\sigma'_{\rm t}, \sigma'_{\rm f})\right)\right)\star \mathcal{N}\left(.;(0,0)^T,  \Sigma'(t,f))\right) \right)(\mu(t,f))
\end{align*}
with $\Sigma'(t,f) + diag(
\sigma'_{\rm t}, \sigma'_{\rm f})=\Sigma(t,f)$. (Proof in Appendix~\ref{proof:gaussianproduct}.)
\end{thm}
\qq

From the above we see how computing the K-transform for various kernels can be done from a base representation which is not the WV but the already convolved WV (\ref{eq:fast}) allowing fast computation as we informally highlight below.

\begin{prop}
\label{prop:complexity}
The K-transform time and/or frequency resolution is inversely proportional to its speed of computation.
\end{prop}
\qq

The above result demonstrates how the fast transform can be used without any downside as long as the learned parameters $\theta$ are not too precise in time and/or frequency. It becomes even better if the parameters are shared across time, leading to an equivariant representation. For details on the Gaussian window truncation please see Appendix~\ref{appendix:truncate}.

\subsection{Computational complexity and pseudo-code}
\label{sec:complexity}

We provide below the explicit pseudo code that summarizes all the involved steps and their impact of the final representation obtained, for the computational complexity see Appendix~\ref{appendix:complexity}, we describe below the pseudo code from Thm.~\ref{thm:fast}. First , do a Gabor transform of $\bx$ with a
    Gaussian window $g_{\sigma_{1/\rm f}}$ for apodization leading to $\STFT_{\bx,\sigma_{1/\rm f}}$. Selection of the $1/\sigma_{\rm f}$ parameter will determine the final frequency resolution of the transform given by $\sigma_{\rm f}$ (recall Prop.~\ref{prop:complexity} and Fig.~\ref{fig:example}).
    Second, do the spectral auto-correlation of $\STFT_{\bx,\sigma_{1/\rm f}}$ with a gaussian spectral apodization window $g_{1/\sigma_{\rm t}}$ with spread $1/\sigma_{\rm t}$ to obtain (\ref{eq:fast}). This maintains the frequency resolution of $\STFT_{\bx,\sigma_{1/\rm f}}$ while increasing the original time resolution $\frac{1}{\sigma_{\rm f}}$ by $\sigma_{\rm t}$. Third, compute the $\K$-transform as per Thm.~\ref{thm:fast} by applying the $2D$ kernels $\Phi[t,f]$ onto (\ref{eq:fast}).

\section{Learning and Experiments}

We now propose to briefly describe how learning of the $\K$-transform is done and validate the method on various datasets.

\subsection{Gaussian Parameter Learning}

We propose to learn the parameters $\theta(t,f)$ from (\ref{eq:theta}). Since all the involved operations of the $\K$-transform are differentiable and we will pipe the transform to various DN architectures we leverage gradient based learning. In order to keep an unconsrained optimization problem we instread explicitly constraint the parameters as follows $
    \sigma_{\rm f}(t,f)={\rm abs}(\sigma_{\rm f}'(t,f)) + \epsilon$, 
    $\sigma_{\rm t}(t,f)={\rm abs}(\sigma_{\rm t}'(t,f)) + \epsilon$, and $
    \rho(t,f) = \text{tanh}(\rho'(t,f))\sqrt{\sigma_{\rm t}(t,f)\sigma_{\rm f}(t,f)}$. With this parametrization, the unconstrained learnable parameters are  $\sigma_{\rm f}'(t,f),\sigma_{\rm t}'$ and $\rho'$. This also allows to always have a nonzero covariance determinant since $\sigma_{\rm time}>0$ and $\sigma_{\rm freq}>0$ and ${\rm det}(C(t,f)) =\sigma_{\rm time}\sigma_{\rm freq}-\rho^2 > 0$. We can now train those unconstrained parameters along with any other model parameters with some flavors of gradient descent for any task at hand. In our case we focus on classification.

\subsection{Experimental Validation}
\label{sec:experiments}

We propose to validate our method on various classification tasks we briefly describe below. We briefly describe below the used dataset and DN architectures and provide the accuracy results averaged over $10$ runs and over multiple learning rates in Table~\ref{tab:my_label}

{\bf DOCC10.}~
The goal of this dataset is to classify recordings obtained from various subsea acoustic stations or autonomous surface vehicles (ASV). The recordings consist of clicks and the task is to perform classification based on the corresponding emitting specy among $10$ marine mammals such as Grampus griseus- Risso's dolphin, Globicephala macrorhynchus
This task is published by ENS Paris College de France: 
%\richb{save space and make the 2 footnotes into citations}
DOCC10 \footnote{\url{https://challengedata.ens.fr/challenges/32}}, for additional details please see Appendix~\ref{sec:docc10}.
{\bf Audio-MNIST.}~
This dataset consists of multiple ($60$) speakers of various characteristics enunciation digits from $0$ to $9$ inclusive \cite{becker2018interpreting}. The classification task consists of classifying the spoken digit, $30,000$ recordings are given in the dataset. Each recording is from a controlled environment without external noise except breathing and standard recording device artifacts. 
{\bf BirdVox-70k.}~
This dataset consists of avian recording obtained during  the  fall  migration  season  of  2015. The recording were obtained in North American and is made of a balanced binary classification of predicting the presence or not of a bird in a given audio clip 
\cite{lostanlen2018birdvox}. In total $70,000$ clips are contaiend in the dataset. The outside recordings present various types of non stationnary noises and sources.
{\bf Google Speech Commands.}~
This dataset consists of $65,000$ one-second long utterances of $30$ short words such as ``yes'', ``no'', ``right'' or ``left'' \cite{speechcommands}. The recordings are obtained from thousands of different people who contributed through the AIY website \footnote{\url{https://aiyprojects.withgoogle.com/open_speech_recording}}. The task is thus to classify the spoken command.
{\bf FreeSound DCASE.}~ This task addresses the problem of general-purpose automatic audio tagging \cite{fonseca2018general}. The dataset provided for this task is a reduced subset from FreeSound and is a large-scale, general-purpose audio dataset annotated with labels from the AudioSet Ontology. There are $41$ classes and $\approx$95,000 training samples, classes are diverse such as Bus, Computer keyboard, Flute, Laughter or Bark. The main challenge of the task is the noise present in the training set labels reflecting the expensiveness of having high quality annotations.

For each dataset we experiment composing the TFR with three DN architectures (for detailed description of those architectures and additional hyper parameters please see Appendix~\ref{appendix:architecture}). In short, we experiment with TFR-time pooling-linear classifier ({\bf Linear Scattering}), TFR-time pooling-MLP ({\bf Nonlinear Scattering}) and TFR-Conv 2D - time pooling - linear classifier ({\bf Linear Joint Scattering}). Those terminologies are based upon the Scattering Network as they perform time averaging \cite{bruna2013invariant}.

We compare our model that we abbreviate as lwvd to the learned Morlet filter-bank \cite{NIPS2018_7711} denoted as lmorlet, and to the learnable sinc based wavelet \cite{ravanelli2018interpretable} denoted as lsinc, which are the current state-of-the-art techniques proposing a learnable TFR. In order to calibrate all the results we also compare with a fixed Morlet filter-bank which is the one that is often seen as the most adapted when dealing with speech and bird signals. Here are external benchmarks, for Audio MNIST using AlexNet on spectrogram leads to $95\%$ \cite{becker2018interpreting} accuracy; for bird Vox, using $3$ layers of 2D convolution and $2$ fully connected layers on top of a fixed a priori designed TFR leads to $90.48\%$ test accuracy \cite{lostanlen2018birdvox}; for the Google command dataset, a DenseNet 121 without pretraining nor data augmentation reaches $80\%$ test accuracy \cite{de2018neural}. In all cases our model reaches comparable performanes while relying on a much simpler DN, and outperforms other learnable TFRs and the a priori optimal one across dataset and optimization settings, offering significan performance gains. We propose in Appendix~\ref{appendix:figures} all the figures of the learned filters and their interpretation.

\begin{table*}[t]
\caption{
    {\small 
    Average over $10$ runs of classification result using three architectures, one layer scattering followed by a linear classifier (Linear Scattering), one layer scattering followed by a two layer neural network (Nonlinear Scattering) and a two layer scattering with joint (2D) convolution for the second layer followed by a linear classifier (Linear Joint Scattering). For each architecture and dataset, we experiment with the baseline, Morlet wavelet fitler-bank (morlet), and learnable frameworks being ours (lwvd), learnable sinc based filters (sinc) and learnable Morlet wavelet (lmorlet). As can be seen, across the dataset, architectures and learning rates, the proposed method provides significant performance gains.}}
    \centering
\setlength{\tabcolsep}{2.2pt}
    \def\arraystretch{1.14}%
\begin{tabular}{|l|r||r||r|r|r||r||r|r|r||r||r|r|r|}
\hline
\multicolumn{1}{|c}{}& & \multicolumn{4}{c||}{Linear Scattering}&\multicolumn{4}{c||}{Nonlinear Scattering}&\multicolumn{4}{c|}{Linear Joint Scattering}\\
\multicolumn{1}{|c}{}&\multicolumn{1}{c||}{le. rate}&{\em morlet}&lwvd&sinc&lmorlet&{\em morlet}&lwvd&sinc&lmorlet&{\em morlet}&lwvd&sinc&lmorlet\\ \hline\hline
 \multirow{3}{*}{\rotatebox{90}{DOCC10}} 
 & 0.0002 &14.3 & 63   & 31.1 & 29.7 & 54.1 & 84.7 & 74.4 & 74.9 & 70.7 & {\bf 83.7} & 82.4 & 75.8 \\
 & 0.001  &12.7 & 65.5 & 26.0   & 28.3 & 50.1 & {\bf 87.9} & 77.4 & 77.4 & 70.1 & 80.6 & 80.8 & 73.2 \\
 & 0.005  &13.0   & {\bf 65.9} & 17.1 & 27.0   & 51.8 & 87.1 & 43.3 & 83.2 & 65.9 & 78.0   & 70.5 & 80.8 \\
\hline\hline
 \multirow{3}{*}{\rotatebox{90}{BirdVox}}
  & 0.0002 &63.8 & 77.9 & 69.6 & 65.4 & 84.7 & 92.9 & 88.1 & 85.8 & 82.1 & {\bf 90.5} & 87.2 & 84.3 \\
 & 0.001  &65.0   & 80.0   & 67.2 & 64.3 & 85.0   & {\bf 94.2} & 88.1 & 86.6 & 80.3 & 88.7 & 86.8 & 83.1 \\
 & 0.005  &65.2 & {\bf 80.4} & 67.3 & 66.9 & 84.8 & 94.2 & 86.0   & 87.2 & 78.1 & 87.5 & 78.3 & 82.8 \\
\hline\hline
 \multirow{3}{*}{\rotatebox{90}{MNIST}}
 &0.0002&43.9 & 68.4 & 52.2 & 44.0   & 82.3 & 85.3 & 10.4 & 83.0 & 95.3 & 97.6 & 22.1 & 95.4 \\
 &0.001&41.5 & {\bf 68.8} & 43.5 & 42.2 & 83.2 & {\bf 89.8} & 87.1 & 85.4 & 89.7 & {\bf 97.8} & 93.2 & 90.4 \\
 &0.005&34.6 & 68.8 & 23.9 & 36.0   & 82.7 & 22.1 & 68.7 & 88.1   & 81.1 & 12   & 64.4 & 80.2 \\
 \hline \hline
 \multirow{3}{*}{\rotatebox{90}{command}}
 &0.0002&8.1 & 24.9 & 9.5 & 7.6 & 33.9 & 38.2 & 36.2 & 33.4 & 65.8 & {\bf 76.7} &  3.7 & 66.8 \\
 &0.001&7.5 & {\bf 26.1} & 8.0   & 8.2 & 33.5 & {\bf 42.9} & 35.5 & 33.7 & 53.6 & 71.8 & 27.9 & 51.9 \\
 &0.005&7.3 & 25.7 & 6.2 & 6.5 & 33.0   & 17.0   & 28.9 & 34.8 & 32.1 & 35.2 & 17.2 & 32.9 \\
 \hline \hline
 \multirow{3}{*}{\rotatebox{90}{fsd}}
 
  &0.0002&9.7 & 15.3 & 10.3 &  9.0   & 22.9 & 23.1 &  2.3 & 27.9 & 40.1 & 38.8 &  1.6 & 42.0   \\
  &0.001&9.8 & 16.7 & 10.4 & 10.6 & 24.2 & 27.4 & 13.1 & {\bf 31.1} & 38.9 & {\bf 44.9} &  2.1 & 42.3 \\
 &0.005&9.0   & {\bf 17.4} &  5.5 & 10.2 & 24.2 & 28.8 & 16.9 & 30.4 & 25.0   & 31.5 & 17.0   & 33.2 
 \\ \hline
\end{tabular}
    \label{tab:my_label}
\end{table*}

We provide the results in Tab.~\ref{tab:my_label} averaged over $10$ runs, for each of the runs, the same data split, DN initialization and parameters are used across all TFRs to allow exact performance comparisons. We also provide the results across various learning rate to perceptually measure the sensitivity of each method to this parameter. The first key observation is that the learnable methods are much more sensitive to the learning rate than when using a fixed TFR. Nevertheless, the proposed $\K$-transform is able to outperform all methods across the datasets and for any DN. This comes from the extrem adaptivity of the produced TFR. Notice that the fixed morlet TFR reaches reasonable accuracy especially on speech data without noise (audio MNIST). This is another key feature of learnable TFR, the ability to learn more robust representations. Another key observation comes from the ability of the proposed method to reach state-of-the-art performances while leveraging a simple (few layer) DN in the Linear Joint Scattering case.

%% file: conclusions.tex
\section{Conclusions}
\label{sec:conc}

We proposed a novel approach to learn generic WVD based TFR, derived an efficient implementation and demonstrated its ability to outperform standard and other learnable TFR techniques across dataset and architecutre settings. 
In addition of learning any desired TFR, our framework is interpretable, telling what type of TFR is, and access to direct parameters allow to exactly position a learned transform among the standard ones.
%This brings new overview for the classification of biological transients# (EEG, ECG), sismic waves and other complex cases that are not yet completely solved. Learnt WVD will in future work be extended to image classification after FFT2D or 2D scalogram.

%s shall open doors to new constraints on DN layers to maximize the number of regions and avoid degenerate region shapes.

%% file: proof.tex
\Huge
\begin{center}
Supplementary Material    
\end{center}
\normalsize

This appendix proposes to first review the implementation details and visual resutls of the paper, studying the learned filters, we conclude with all the proofs of the theoretical results.

\section{DN topology}
\label{appendix:architecture}

We leverage a time translation covariant form of the proposed learnable model as we aim at solving classification tasks based on audio clips. Hence the representation should be translation invariant. We keep the unconstrained frequency dimensions and thus do not impose any frequency shift invariance as in the Cohen class family of representations. The kernels are parametrized as given in the main text and the networks are given as follows:
       \begin{verbatim}
       TF: any representation (morlet, lwvd, ...)
       the first mean(3) represents time pooling

       - onelayer_nonlinear_scattering
            input = T.log(TF.mean(3).reshape([N, -1])+0.1) 
            Dropout(0.3)
            Dense(256)
            BatchNormalization([0])
            LeakyReLU
            Dropout(0.1)
            Dense(n_classes)

        - onelayer_nonlinear_scattering:
            input = T.log(TF.mean(3).reshape([N, -1])+0.1)
            Dropout(0.1)
            Dense(n_classes)
        
        - joint_linear_scattering:
            feature = T.log(TF.mean(3).reshape([N, -1])+0.1)
            
            input = T.log(TF+0.1)
            Conv2D(64, (32,16))
            BatchNormalization([0,2,3])
            AbsoluteValue
            Concatenate(AbsoluteValue, feature)
            Dropout(0.1)
            Dense(n_classes)
       \end{verbatim}           
all training are done with the Adam optimizer, same initialization and data splitting.

\section{Additional Figures}
\label{appendix:figures}

We represent in this section the leaned filters/kernels $\Phi$ applied on the smoothed pseudo Wigner-Ville distribution, for clarity we only depict one every 4 filters, concatenated horizontally. We do so for three dataset and provide analysis in the caption of each figures.

\subsection{Samples of learnt filters}

We propose in Fig.~\ref{fig:extra} and Fig.~\ref{fig:extra2} and Fig.~\ref{fig:extra3} the filters after learning for each dataset with their analysis.

\begin{figure}[ht]
    \centering
    \includegraphics[width=0.9\linewidth]{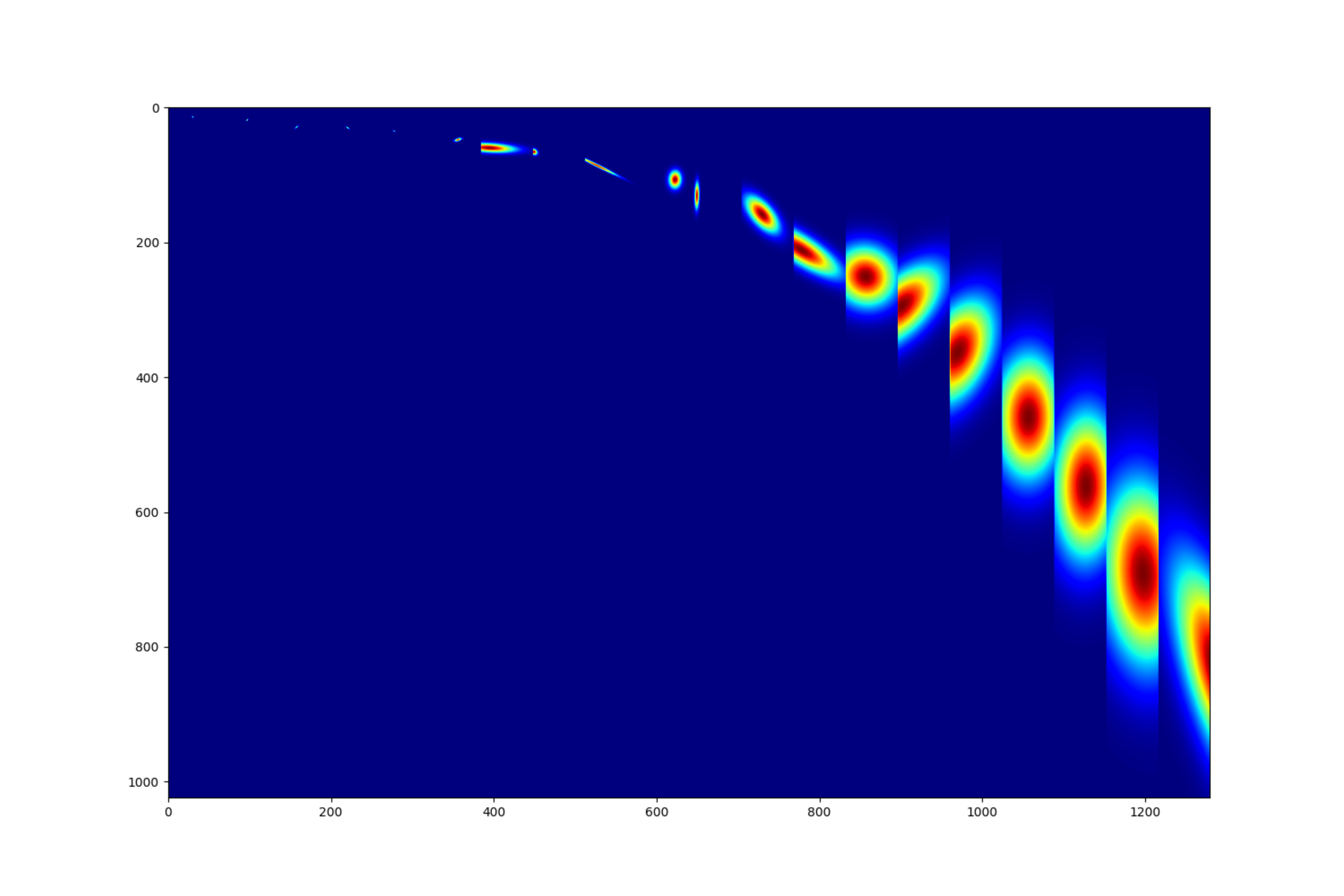}
    \caption{Audio MNIST: this dataset deals with spoken digit classification. A few key observations: the high frequency filters tend to take an horizontal shape greatly favoring frequency resolution, for some of the filters the time resolution is also very high (reaching super-resolution) while others favor local translation invariance. For all the medium to low frequency filters, great frequency and time invariance is preferred (large gaussian support) with a slight chirpness for the medium frequency filters. The low frequency filters tend to favor time resolution.}
    \label{fig:extra}
\end{figure}
\begin{figure}[ht]
    \centering
    \includegraphics[width=0.9\linewidth]{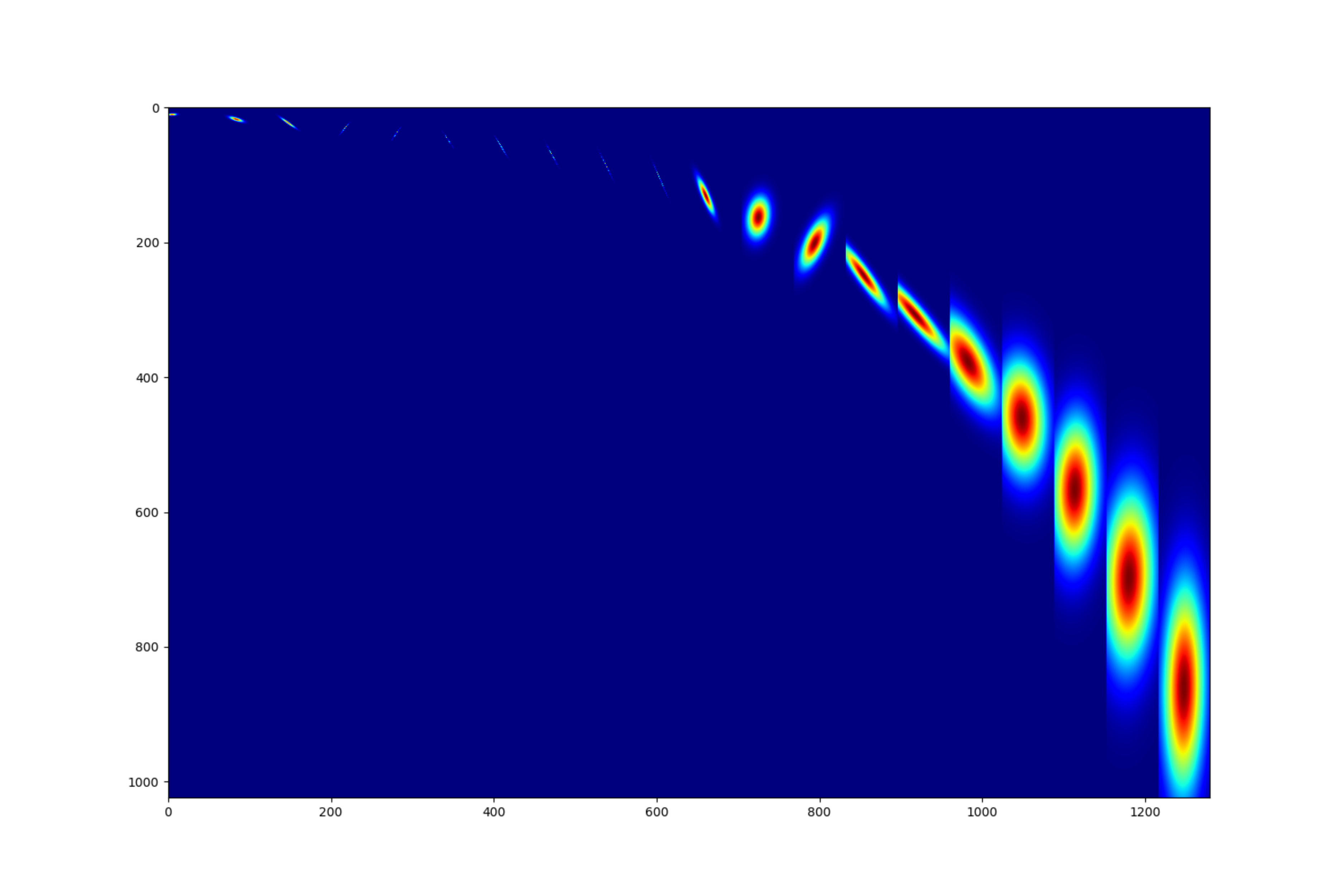}
    \caption{FreeSound: this dataset contains various different classes ranging on different frequencies and without an a priori prefered form of the events in the WV space. As opposed to the AudioMNIST case, we can see that the kernels tend to have smaller covariance (support) hence preferring time and frequency resolution to invariance. This becomes especially true for the high frequency atoms. We also see the clear chirpness for the medium/high frequency kernels with a specific (-30) angle, with decreasing slope. This might be specific to some particular events involving moving objects such as train, cars and so on.}
    \label{fig:extra2}
\end{figure}
\begin{figure}[ht]
    \centering
    \includegraphics[width=0.9\linewidth]{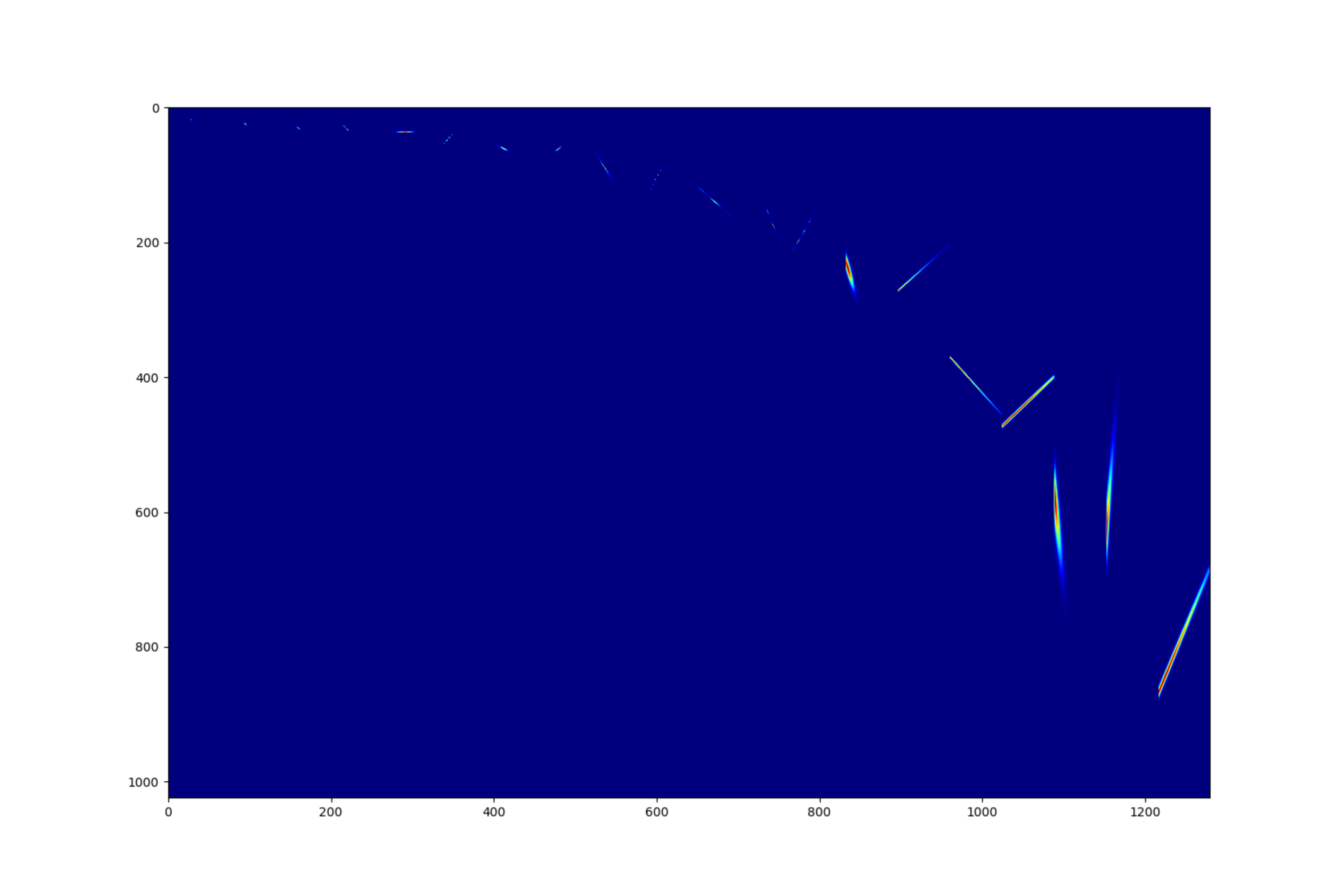}
    \caption{Bird: This dataset proposes to predict the presence or absence of a bird in short audio clips. A priori, detection of such events heavily relies on chirps, the characteristic sound of birds with increasing/decreasing frequency over time. We can see from the learned filters how the learned representation indeed focuses on such patterns, reach super-resolution and thus extreme sensitivity to the time and frequency position of the events, in particular for medium to low frequency kernels. For high frequency kernels, the time and frequency resolution is further increased.}
    \label{fig:extra3}
\end{figure}

\section{Proofs}

In this section we present in details all the proofs of the main paper results, as well as providing some additional ones.

\subsection{Proof of Cor.~\ref{cor:lipschitz_theta}}
\label{proof:lipschitz_theta}

Let first prove the general case with arbitrary kernels
\begin{lemma}
\label{lemma:lipschitz}
The norm of the difference of two representations obtained from kernel $\Phi$ and $\Phi'$ is bounded above as $
    \| \K_{x,\Phi} - \K_{x,\Phi'} \|_{L^2(\mathbb{R}^2)} \leq \|x \|^2_{L^2(\mathbb{R})}  \|\Phi - \Phi' \|_{L^2(\mathbb{R}^4)}.
$
with $\|\Phi - \Phi' \|_{L^2(\mathbb{R}^4)}= \sqrt{\int_{t,f}\|(\Phi - \Phi')[t,f]\|^2_{L^2(\mathbb{R}^2)}}$.
\end{lemma}

\begin{proof}
First, one can easily derive  $\|\WV_{x}\|_{L^2(\mathbb{R})}= \|x \|_{L^2(\mathbb{R}^2)}^2$ (see for example \cite{cordero2018sharp}). Given this, and the definition of the $\K$-transform, we obtain that
\begin{align*}
    \| \K_{x,\Phi}-\K_{x,\Phi'} \|_{L^2(\mathbb{R}^2)} = &\sqrt{ \int_{t,f} (\langle \WV_{x}, \Phi[t,f]\rangle_{L^2(\mathbb{R}^2)} - \langle \WV_{x}, \Phi'[t,f]\rangle_{L^2(\mathbb{R}^2)})^2}\\
    =&\sqrt{ \int_{t,f} \langle \WV_{x}, \Phi[t,f]-\Phi'[t,f]\rangle_{L^2(\mathbb{R}^2)}^2}\\
    \leq &\sqrt{ \int_{t,f} \| \WV_{x}\|^2_{L^2(\mathbb{R}^2)}\| \Phi[t,f]-\Phi'[t,f]\|^2_{L^2(\mathbb{R}^2)}}\;\;\;\text{Cauchy-Schwarz Inequality}\\
    =&\| \WV_{x}\|_{L^2(\mathbb{R}^2)}\sqrt{ \int_{t,f} \| \Phi[t,f]-\Phi'[t,f]\|^2_{L^2(\mathbb{R}^2)}}\\
    =&\| x\|^2_{L^2(\mathbb{R})}\sqrt{ \int_{t,f} \| \Phi[t,f]-\Phi'[t,f]\|^2_{L^2(\mathbb{R}^2)}}\\
    =&\| x\|^2_{L^2(\mathbb{R})} \| \Phi-\Phi'\|_{L^2(\mathbb{R}^4)}
\end{align*}
\end{proof}

Now the proof for the special case of a 2D Gaussian kernel follows the exact same procedure as the one above for the Lipschitz constant of the general $\K$-transform.
\begin{proof}
In the same way we directly have
\begin{align*}
    \| \K_{x,\Phi}-\K_{x,\Phi'} \|_{L^2(\mathbb{R}^2)} \leq &\| x\|^2_{L^2(\mathbb{R})}\sqrt{ \int_{t,f} \| \Phi[t,f](\theta)-\Phi[t,f](\theta')\|^2_{L^2(\mathbb{R}^2)}}
\end{align*}
now the only different is that we need a last inequality to express the distance in term of the $\theta$ parameters and not the kernels $\Phi$. To do so, we leverage the fact that the parametric kernel is a $2$-dimensional Gaussian with zero mean and unit variance and leverage the following result:
\begin{align*}
    \|f(\theta)-f(\theta') \| \leq \max_{u} \|\nabla_f(u)\| \|\theta - \theta' \|
\end{align*}
and simply denote by $\kappa$ the maximum of the Gaussian gradient norm, leading to the desired result by setting $f$ the $2$-dimensional Gaussian.
\end{proof}

\subsection{Proof of Prop.~\ref{prop:universal_gabor}}
\label{proof:universal_gabor}

First we prove the more general result which follows.

\begin{lemma}
\label{lemma:universal}
For any TFR or signal adapted TFR, there exists a kernel $\Phi$ such that $\K_{x,\Phi}$ (recall (\ref{eq:transform})) is equal to it.
\end{lemma}
\qq

\begin{proof}
The proof is a direct application of Moyal Theorem which states that given a signal $x$ and a filter $y$ the application of the filter onto the signal and squaring the result can be expressed as the inner product between the Wigner-Ville transforms of the signal and filter as in
\begin{align*}
    |\int_{-\infty}^{\infty}x(t)y^*(t)dt |^2 = \int_{-\infty}^{\infty}\int_{-\infty}^{\infty}\WV_x(\tau, \omega)\WV_{y}(\tau,\omega) d\tau d\omega.
\end{align*}
From the above we can see that any time frequency representation (time invariant or not) can be recovered simply by setting $\Phi[t,f]=\WV_{y}$ of some desired filter $y$.
\end{proof}

Now for the special case of the Gabor based transform we can leverage the above result and proof with the following.

\begin{proof}
This result is a direct application of the analytically derived kernels from \cite{nuttall1988wigner,flandrin1990affine,jeong1990variable,baraniuk1996wigner,talakoub2010approximating} for various analytical time-frequency representations. For the Gabor transform, Morlet wavelet and Morlet based Chirplet atom, all the analytically derived kernels that are applied onto a Wigner-Ville distribution as a parametric form of the $2$-dimensional Gaussian kernel proving the result.
\end{proof}
\iffalse
\subsection{Proof of Theorem~\ref{thm:logon}}
\label{proof:logon}
\begin{proof}
\url{http://scipp.ucsc.edu/~haber/ph116A/diag2x2_11.pdf}
\end{proof}
\fi

\subsection{Proof of Theorem~\ref{thm:continuity}}
\label{proof:continuity}
\begin{proof}
The proof follows directly from the above result. In fact, we obtained above that the Lipschitz constant of the $\K$-transform w.r.t. the $\Phi$ parameter is obtained by $\|x\|^2_{L^2(\mathbb{R})}$. Thus, the Lipschitz continuity implies continuity of the transform w.r.t. the $\Phi$ parameters which proves the result (see for example \cite{kreyszig1978introductory}).
\end{proof}

\subsection{Proof of Lemma on Invariant}
%~\ref{lemma:invariant}}
\label{proof:invariant}
\begin{proof}

We provide the proof for each case below:
\begin{itemize}

    \item Time translation of the signal $x$ by $\tau$ to have $y(t) = x(t-\tau)$ leads to
    \begin{align*}
        \WV_{y}(t,\omega)&=\int_{-\infty}^{\infty}y(t+\frac{\tau}{2})y^*(t-\frac{\tau}{2})e^{-i\tau\omega}d\tau
        \\
        &=\int_{-\infty}^{\infty}x(t+\frac{\tau}{2}-\tau)x^*(t-\frac{\tau}{2}-\tau)e^{-i\tau\omega}d\tau
        \\
        &=\WV_{x}(t-\tau,\omega)
    \end{align*}
    Thus, the time equivariance of the representation defined as $\K_{y,\Phi}(t,f)=\K_{x,\Phi}(t-\tau,f)$ is obtained iff
    \begin{align*}
        \K_{y,\Phi}(t,f)=&\langle \WV_{y}, \Phi[t,f]\rangle
        =\langle \WV_{x}, \Phi[t,f](.+\tau,.)\rangle=\langle \WV_{x}, \Phi[t-\tau,f]\rangle \\
        &\iff \Phi[t,f](.+\tau,.) = \Phi[t-\tau,f]
    \end{align*}
    as a result the above condition demonstrates that filters of different times $\Phi[t-\tau,f], \forall \tau $are just time translations of each other giving the desired result.

    \item Frequency modulation/shift with frequency $\omega_0$ to have $y(t) = x(t)e^{i\omega_0 t}$ leads
    \begin{align*}
        \WV_{y}(t,\omega)&=\int_{-\infty}^{\infty}y(t+\frac{\tau}{2})y^*(t-\frac{\tau}{2})e^{-i\tau\omega}d\tau\\
        &=\int_{-\infty}^{\infty}x(t+\frac{\tau}{2})e^{i\omega_0(t+\frac{\tau}{2}) }x^*(t-\frac{\tau}{2})e^{-i\omega_0(t-\frac{\tau}{2}) }e^{-i\tau\omega}d\tau\\
        &=\int_{-\infty}^{\infty}x(t+\frac{\tau}{2})e^{i\omega_0\tau) }x^*(t-\frac{\tau}{2})e^{-i\tau\omega}d\tau=W_{x}(t-\tau,\omega)\\
        &=\int_{-\infty}^{\infty}x(t+\frac{\tau}{2})x^*(t-\frac{\tau}{2})e^{-i\tau(\omega+\omega_0)}d\tau=W_{x}(t-\tau,\omega)\\
        &=\WV_{x}(t,\omega+\omega_0)
    \end{align*}
    now leveraging the same result that for the time equivariance, we obtain the desired result.
    
    \item For completeness, we also demonstrate here how the Wigner-Ville behaves under rescaling of the input $y(t)=x(t/a)$ 
    \begin{align*}
        \WV_{y}(t,\omega)&=\int_{-\infty}^{\infty}y(t+\frac{\tau}{2})y^*(t-\frac{\tau}{2})e^{-i\tau\omega}d\tau\\
        &=\int_{-\infty}^{\infty}x((t+\frac{\tau}{2})/a)x^*((t-\frac{\tau}{2})/a)e^{-i\tau\omega}d\tau\\
        &=\int_{-\infty}^{\infty}x(t/a+\frac{\tau}{2a}))x^*(t/a-\frac{\tau}{2a})e^{-i\tau\omega}d\tau\\
        &=\int_{-\infty}^{\infty}x(t/a+\frac{\tau}{2a}))x^*(t/a-\frac{\tau}{2a})e^{-i\tau\omega}d\tau\\
        &=a\int_{-\infty}^{\infty}x(t/a+\frac{\tau}{2}))x^*(t/a-\frac{\tau}{2})e^{-i\tau a\omega}d\tau\\
        &=a \WV_{x}(t/a,a\omega)
\end{align*}
\end{itemize}
\end{proof}

\subsection{Proof of Prop~\ref{prop:invariance}}
\label{proof:invariance}
\begin{proof}
The proof is obtained by using the fact that the maximum of the $2$-dimensional Gaussian with given covariance matrix will always be smaller than $\frac{1}{\det (\Sigma)}$ and is always nonnegative, and then simply upper bounding the norm difference by this value times the representation. Now that we have the upper bound in term of $\|WV_{x}-WV_{D(X)} \|_{L^2(\mathbb{R}^2)}$ simply apply the Lipschitz constant inequality with $\kappa$ the constant of the WVD which exists as long as $x$ is bounded which is the case as we have $\mathbb{L}^2(\mathbb{R})$ leading to the desired result.
\end{proof}

\subsection{Proof of Lemma~\ref{thm:fast}}
\label{proof:fast}

\begin{proof}
\iffalse
    \begin{align*}
        \int_{-\infty}^{\infty}& ST_{x}(t,\omega+\frac{\eta}{2})ST_{x}^*(t,\omega-\frac{\eta}{2})e^{j2 \pi \eta t}d\eta\\
        =&\int_{-\infty}^{\infty} \int_{-\infty}^{\infty}  w^*(t-\tau)x(\tau)e^{-i \tau (\omega+\frac{\eta}{2})}d\tau
        \int_{-\infty}^{\infty} w(t-\theta)x^*(\theta)e^{i \theta (\omega-\frac{\eta}{2})}d\theta
         e^{j2 \pi \eta t}d\eta\\
        =&\int_{-\infty}^{\infty} \int_{-\infty}^{\infty} w^*(t-\tau)x(\tau)w(t-\theta)x^*(\theta)e^{-i \omega(\tau-\theta)}\int_{-\infty}^{\infty} e^{j2 \pi \eta (t-\frac{\tau}{2}-\frac{\theta}{2})}d\eta d\tau d\theta\\
        =&\int_{-\infty}^{\infty} \int_{-\infty}^{\infty} w^*(t-\tau)x(\tau)w(t-\theta)x^*(\theta)e^{-i \omega(\tau-\theta)}\delta(t-\frac{\tau}{2}-\frac{\theta}{2})d\tau d\theta \\
        =&\int_{-\infty}^{\infty} w^*(t-\tau)x(\tau)w(\tau-t)x^*(2t-\tau)e^{-i \omega(2\tau-2t)}d\tau\\
        =&\int_{-\infty}^{\infty} w^*(-\eta/2)x(\eta/2+t)w(\eta/2)x^*(t-\eta/2)e^{-i \omega \eta}d \eta /2\;\;\;\;\;\; \eta=2(\tau-t)\\
        =&(W_{w}(0,.) \star W_{x}(t,.))(\omega)/2 \implies \Pi(\tau,\omega) = \delta(\tau) W_{w}(\tau,\omega)/2
    \end{align*}
\fi
We will obtain the desired result by unrolling the following equations and see that it coincides with the one of the Theorem with the given kernel applied on the WV representation:
    \begin{align*}
        \int_{-\infty}^{\infty}& \ST_{x}(t,\omega+\frac{\eta}{2})\ST_{x}^*(t,\omega-\frac{\eta}{2})\mathcal{F}_{g}(\eta)e^{j2 \pi \eta t}d\eta\\
        =&\int_{-\infty}^{\infty} \ST_{x}(t,\omega+\frac{\eta}{2})\ST_{x}^*(t,\omega-\frac{\eta}{2})\int_{-\infty}^{\infty}g(\xi)e^{-i\xi\eta}d\xi e^{j2 \pi \eta t}d\eta\\
        =&\int_{-\infty}^{\infty} \int_{-\infty}^{\infty} w^*(t-\tau)x(\tau)w(t-\theta)x^*(\theta)e^{-i \omega(\tau-\theta)}\int_{-\infty}^{\infty}g(\xi)\int_{-\infty}^{\infty} e^{j2 \pi \eta (t-\frac{\tau}{2}-\frac{\theta}{2}-\xi)}d\eta d\xi d\tau d\theta\\
        =&\int_{-\infty}^{\infty} \int_{-\infty}^{\infty}\int_{-\infty}^{\infty}
        w^*(t-\tau)x(\tau)w(t-\theta)x^*(\theta)e^{-i \omega(\tau-\theta)}g(\xi)\delta(t-\frac{\tau}{2}-\frac{\theta}{2}-\xi)d\xi d\tau d\theta \\
        =&\int_{-\infty}^{\infty}\int_{-\infty}^{\infty}
        w^*(t-\tau)x(\tau)w(\tau+2\xi-t)x^*(2t-\tau-2\xi)g(\xi)e^{-i \omega(2\tau-2t+2\xi)}d\xi d\tau \\
        =&\int_{-\infty}^{\infty}g(\xi)\int_{-\infty}^{\infty}
        w^*(\xi-\mu/2)x(t+\mu/2-\xi)w(\xi+\mu/2)x^*(t-\mu/2-\xi)e^{-i \omega(2\mu)}d\tau d\xi\\
        &\hspace{10cm}\;\;\;(\mu=2(\tau-t+\xi))  \\
        =&\int g(\xi)(W_{w}(\xi,.) \star W_{x}(t-\xi,.))(\omega)/2d\xi \implies \Pi'(\tau,\omega) = g(\tau) W_{w}(\tau,\omega)/2
    \end{align*}
with $\Pi'(t,\omega)=\frac{1}{\sqrt{2\pi}\sigma} e^{-\frac{t^2}{\sigma_{t}^2/(\sigma_{\omega}^2\sigma_{t}^2+1)}-\frac{\omega^2\sigma_{t}^2}{2}}$.

\end{proof}

\subsection{Proof of Thm.~\ref{thm:gaussianproduct}}
\label{proof:gaussianproduct}

\begin{proof}
From the above (lemma) result, it then follows directly that performing convolutions with two gaussians can be rewritten as a single Gaussian convolution with the given parameters based on both Gaussians.
\end{proof}

\section{Interferences}
\label{sec:interference}

Finally, we now consider the study of interference that can arise in the K-transform whenever multiple events occur in the signal $x$ \cite{sen2014uncertainty}.
In fact, the Gabor limit also known as the Heisenberg Uncertainty principle \cite{heisenbard1927}, corresponds to the optimal limit one can achieve in time and frequency resolution without introducing interference in the representation. In our case, no constraints are imposed onto $\Phi[t, f]$, and while learning will allow to reach any desired kernel for the task at hand, we propose some conditions that would prevent the presence of interference. In fact, there is a general necessary condition ensuring that the representation does not contain any interference.

\begin{prop}
\label{prop:interference}
A sufficient condition to ensure absence of interference in the K-transform is to have a nonnegative representation as in $\K_{x,\Phi}(t,f) \geq 0 ,\forall t,f$.
\end{prop}
\qq

\begin{proof}
We only consider signals that can be expressed as a linear combination of oscillatory atoms (natural signals) \cite{mallat1989theory}. Due to the form of the Wigner-Ville distribution (with the quadratic term) the interference are oscillatory terms and thus can not be positive only. In short, there can not be a positive coefficient appearing due to an interference without its negative counterpart. As a result, a nonnegative representation can not have any interference. See (\cite{cohen1995time}) for more details and background on Wigner-Ville distribution and interference. 
\end{proof}

Lastly, thanks to the Gaussian form, we can directly obtain the shape and in particular area of the logon (the inverse of the joint time and frequency resolution) as follows.

\begin{prop}
\label{lemma:interference_area}
A sufficient condition to ensure absence of interference in the K-transform is to have the effective area of the 2D Gaussian greater that $\frac{1}{4\pi}$, that is, $\epsilon_{\rm T}'\epsilon_{\rm F}' \geq \frac{1}{4\pi}, \forall t,f$  with
\begin{align}
\sigma_{\rm T}' =& \sigma_{\rm T} \cos^2(\theta) + 2\rho \cos(\theta)\sin(\theta) + \sigma_{\rm F}\sin^2(\theta) \\
\sigma_{\rm F}' =&\sigma_{\rm T}\sin^2(\theta)-2\rho \cos(\theta)\sin(\theta)+\sigma_{\rm F}\cos^2(\theta),\\ 
\theta =& \frac{\arctan(\frac{2 \rho}{\sigma_{\rm T}-\sigma_{\rm F}})}{2}
\end{align}
\end{prop}
\begin{proof}
The proof applies the Uncertainty principle which provides the minimal area of the Logon to ensure absence of interference (\cite{gabor1946theory}), this is then combined with the above result on the area of the Logon in the case of our $2$-dimension Gaussian to obtain the desired result. So first, the condition is that $\sigma_{\rm T}\sigma_{\rm F}\geq \frac{1}{4\pi}$, however in our case we also have the possible chirpness parameter. But we can obtain the new rotated covariance matrix s.t. the chirpness is removed by rotating the time-frequency place instead. This can be done easily (see for example \url{http://scipp.ucsc.edu/~haber/ph116A/diag2x2_11.pdf}) to obtain the new parameter as
\begin{align*}
    \sigma_{\rm T} =& \sigma_{\rm T} \cos^2(\theta) + 2\rho \cos(\theta)\sin(\theta) + \sigma_{\rm F}\sin^2(\theta) 
    \\ \sigma_{\rm F}=&\sigma_{\rm T}\sin^2(\theta)-2\rho \cos(\theta)\sin(\theta)+\sigma_{\rm F}\cos^2(\theta),\\ 
    \theta =& \frac{\arctan(\frac{2 \rho}{\sigma_{\rm T}-\sigma_{\rm F}})}{2}
\end{align*}
now the constraint can be expressed as
\begin{align*}
    \left(\sigma_{\rm T} \cos^2(\theta) + 2\rho \cos(\theta)\sin(\theta) + \sigma_{\rm F}\sin^2(\theta) \right)\left(\sigma_{\rm T}\sin^2(\theta)-2\rho \cos(\theta)\sin(\theta)+\sigma_{\rm F}\cos^2(\theta)\right)\geq \frac{1}{4\pi}
\end{align*}
\end{proof}

\section{Gaussian Truncation}
\label{appendix:truncate}

Notice however that  the STFT has to be done by padding the signal windows to allow interpolation, this is common when dealing with such transformations. %\cite{mallat}. 
Given some parameters, we convert them into discrete bins to get actual window sizes as follows
\begin{align*}
N_{\sigma}(\epsilon) = - g^{-1}_{\sigma}(\epsilon) \times 2 Fs
\end{align*}
with $g^{-1}_{\sigma}$ the inverse of the gaussian density distribution with $0$ mean and $\sigma$ standard deviation, taken only on the negative part of its support. As such, $N$ is a function that maps a given tolerance and standard deviation to the window length in bins s.t. the apodization window at the boundaries of this window are of no more than $\epsilon$.

\section{Example of K-transform parameterization}

We propose in Table~\ref{tab:summary_parameters} different configurations of the parameter $\theta$ corresponding to some standard and known TFRs.

\begin{table*}[t]
\caption{\small{Various hand picked K-transform parameters leading to known time invariant TFRs with their respective parameters. For the adaptive versions, such as the wavelet tree with various Gabor mother wavelet parameters $\sigma_0$, then at each time $t$, the corresponding parameter is any of the optimal one based on the desired criterion.}}
\centering
\setlength{\tabcolsep}{2pt}
    \def\arraystretch{1.22}%
\begin{tabular}{|l|c|c|c|c|c|}\cline{1-6}
&$\mu_{\rm time}(t,f)$ & $\mu_{\rm freq}(t,f)$ & $\sigma_{\rm time}(t,f)$ & $\sigma_{\rm freq}(t,f)$ & $\rho(t,f)$ (chirpness) \\ \cline{1-6}
Spectrogram & $t$      &          $f$       &   $\sigma_{t}$   & $\sigma_{t}^{-1} $ & $0$ \\ \hline
Melscale Spectrogram &  $t$      &  $2^{S(1-f/\pi)}$  &   $\sigma_{t}$   & $2^{S(f/\pi-1)}\sigma_{t}^{-1} $ & $0$ \\ \hline
Scalogram & $t$      &          $2^{S(1-f/\pi)}$       &  $2^{S(1-f/\pi)}\sigma_0$ & $ 2^{S(f/\pi-1)}\sigma_0^{-1}$ &    $0$   \\ \hline
Scattering Layer & $t$      &          $2^{S(1-f/\pi)}$       &  $s2^{S(1-f/\pi)}\sigma_0$ & $ 2^{S(f/\pi-1)}\sigma_0^{-1}$ &    $0$   \\ \hline
Chirpogram &  $t$      &          $2^{S(1-f/\pi)}$       &  $2^{S(1-f/\pi)}\sigma_0$ & $ 2^{S(f/\pi-1)}\sigma_0^{-1}$ &    $\rho(t,f) \neq 0$   \\ \hline
\end{tabular}
\label{tab:summary_parameters}
\end{table*}

\section{Computational Complexity}
\label{appendix:complexity}

The computational complexity of the method varies with the covariance matrix of the kernels, it will range from quadratic when the kernels tend toward a Delta function (as the representation computation bottleneck becomes the computation of the exact WV) to linear for large covariance matrices. In such case, the computation is done with spectrograms of small windows (and with almost no spectral frequency correlation) for the smoothed pseudo WV. In addition to those cases, the use of translation (in time and/or frequency) also greatly reduce computational costs as it allows for the use of convolutions into the K-transform computation. This allows to put all the computation bottleneck into the computation of the smoothed pseudo WV.

\section{Data set DOCC10 : Dyni Odontocete Click Classification, 10 species, ENS DATA CHALLENGE}
\label{sec:docc10}

This challenge is linked to the largest international bioacoustical challenge from Scripps Institute and LIS CNRS. The goal is to classify transients, clicks emitted as biosonar by 10 cetaceans species (odontocetes). The recordings of the clicks are from a post-processed subset of DCLDE 2018 challenge (9 species, \url{http://sabiod.org/DCLDE }), plus recordings of Cachalot (Physteter macrocephalus, Pm) done from ASV Sphyrna (\url{http://sphyrna-odyssey.com}) near Toulon, France. Each input signal is 8192 bins, at a sampling rate of 200 kHz. They each includes a click centered in the middle of the window in the case of the test set. The clicks in the training set are in various positions and background noises.

Challenge goals
The goal is to classify each click according to the corresponding emitting species. The 10 species are : (0) Gg: Grampus griseus- Risso's dolphin (1) Gma: Globicephala macrorhynchus- Short-finned pilot whale (2) La: Lagenorhynchus acutus- Atlantic white-sided dolphin (3) Mb: Mesoplodon bidens- Sowerby's beaked whale (4) Me: Mesoplodon europaeus- Gervais beaked whale (5) Pm: Physeter macrocephalus - Sperm whale (6) Ssp: Stenella sp.Stenellid dolphin (7) UDA: Delphinid type A - a group of dolphins (species not yet determined) (8) UDB: Delphinid type B - another group of dolphins (species not yet determined) (9) Zc: Ziphius cavirostris- Cuvier's beaked whale The metric is the MAP.

Data description
The recordings are from various subsea acoustic stations or autonomous surface vehicles (ASV). A part (9 classes) has been provided by the Scripps Institution of Oceanography for the 8th DCLDE Workshop. They consist of acoustic recordings from multiple deployments of high-frequency acoustic recording packages (Wiggins and Hildebrand, 2007) deployed in the Western North Atlantic (US EEZ) and Gulf of Mexico. It has a 100 kHz of bandwidth (200 kHz sample rate). Data were selected to cover multiple seasons and locations while providing high species diversity and call counts as described in \url{http://sabiod.org/DCLDE/challenge.html#highFreqData}. The initial data set is 3 To and weak labeled. DYNI LIS TOULON (Ferrari Glotin 2019) filtered the labels by several detector and clustering, yielding to 90 000 samples (6 Go). Another part (1 class) is from Dyni team recording in Mediterranean sea, from ASV Sphyna, at 384 kHz Fe, downsampled at 200 kHz Fe.  In sum, it yields around 11312 samples per class in the train set, and 2096 samples per class in the test set.

\begin{figure}[H]
    \centering
    \includegraphics[width=0.7\linewidth]{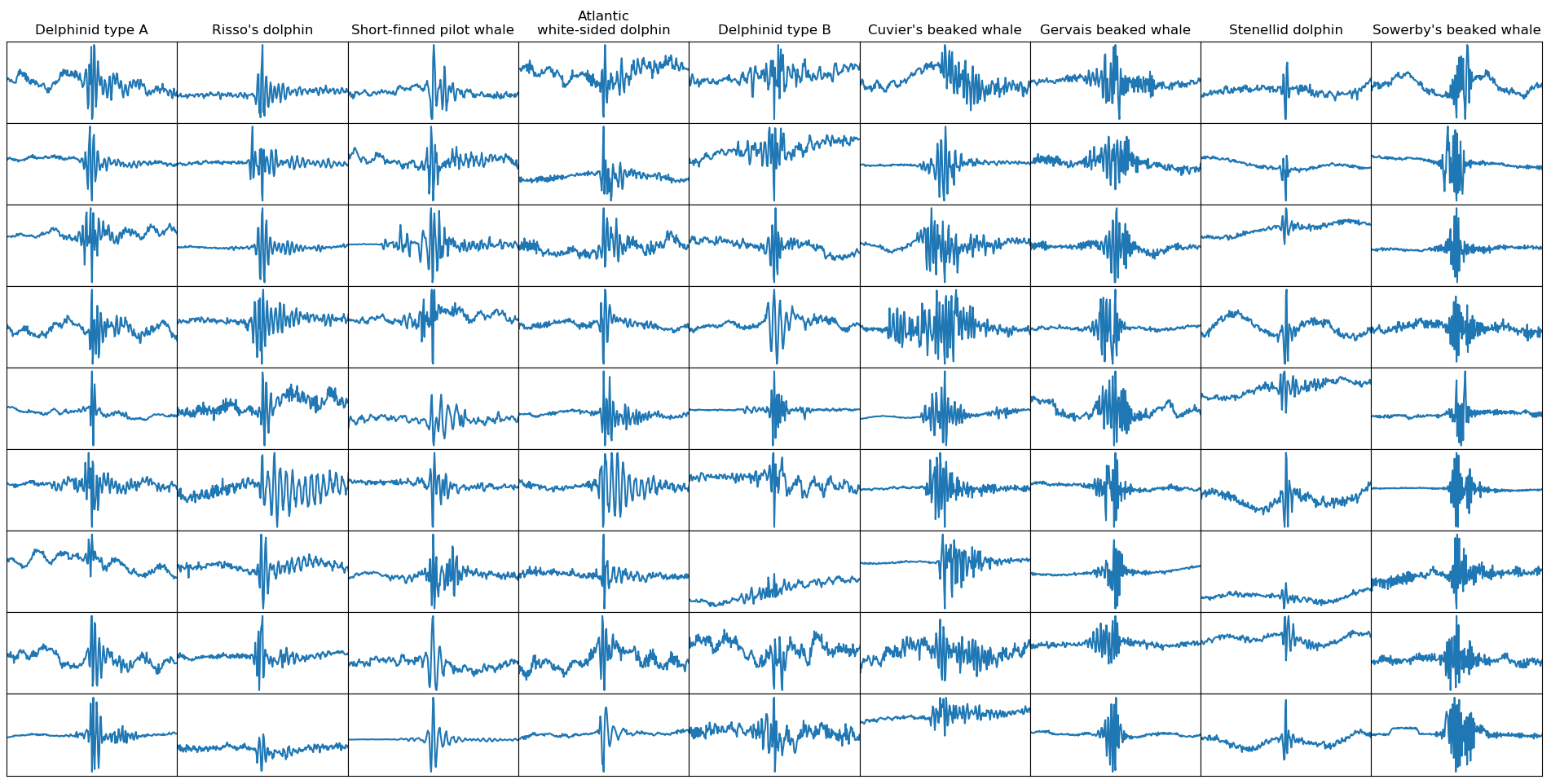}
    \caption{DOCC10 dataset samples, one class per column. The click is centered. Details in ENS DATA CHALLENGE WEB SITE.}
    \label{fig:extra2}
\end{figure}